\documentclass[12pt]{article}

\usepackage[hmargin=2.32cm, vmargin={3cm,3cm}, a4paper]{geometry}\usepackage{amsmath,amssymb,amsthm,mathrsfs}
\usepackage{epsfig,epsf,subfigure,graphicx,graphics}
\usepackage{url,enumerate}
\graphicspath{{fig/}}
\usepackage{float}
\usepackage{parskip}
\usepackage{fancyhdr}
\usepackage[dvipsnames]{xcolor}
\usepackage{physics}
\DeclareMathOperator*{\argmax}{argmax} 
\DeclareMathOperator*{\argmin}{argmin} 
\DeclareMathAlphabet\mathrsfso      {U}{rsfso}{m}{n}
\usepackage{amssymb}\usepackage{mathtools}
\usepackage[toc,page]{appendix}
\usepackage{bbm}

\usepackage{lettrine}
\usepackage{hyperref}

\newtheorem{lemma}{Lemma}
\newtheorem*{lemma*}{Lemma}

\newtheorem{definition}{Definition}
\newtheorem{proposition}{Proposition}
\newtheorem{theorem}{Theorem}
\newtheorem*{theorem*}{Theorem}

\newtheorem{corollary}{Corollary}[theorem]

\begin{document}

\title{On classical advice, sampling advice and complexity assumptions for learning separations}

\author{Jordi Pérez-Guijarro}
\date{\small SPCOM Group, Universitat Politècnica de Catalunya, Barcelona, Spain}

\maketitle

\abstract{
In this paper, we study the relationship between advice in the form of a training set and classical advice. We do this by analyzing the class $\mathsf{BPP/samp}$ and certain variants of it. Specifically, our main result demonstrates that $\mathsf{BPP/samp}$ is a proper subset of the class $\mathsf{P/poly}$, which implies that advice in the form of a training set is strictly weaker than classical advice. This result remains valid when considering quantum advice and a quantum generalization of the training set. Finally, leveraging the insights from our proofs, we identify both sufficient and necessary complexity-theoretic assumptions for the existence of concept classes that exhibit a quantum learning speed-up. We consider both the worst-case setting—where accurate results are required for all inputs—and the average-case setting.}

\section{Introduction}

The study of algorithms with access to advice is crucial in fields such as cryptography, where the aim is to establish protocols that remain secure even against adversaries with advice. However, the origin of these studies has its roots in non-uniform circuits. In particular, Karp and Lipton introduced the class $\mathsf{P/poly}$ in \cite{karp1982turing}. This class denotes the set of decision problems that can be efficiently solved with advice of polynomial-size.

The seminal work of Karp and Lipton \cite{karp1982turing} not only introduced this class but also demonstrated that, although classical advice is quite helpful, it has its limits. This is evident in the fact that classes $\mathsf{NP}$ and $\mathsf{EXP}$ are not contained in $\mathsf{P/poly}$, assuming the conjecture that the polynomial hierarchy does not collapse.
Another crucial result concerning this class is Adleman's theorem \cite{adleman1978two}, which states that randomness is not useful when classical advice is given, i.e., $\mathsf{BPP}$ is a proper subset of $\mathsf{P/poly}$. Interestingly, to encompass quantum algorithms and quantum advice, the class $\mathsf{P/poly}$ was extended, resulting in the definition of $\mathsf{BQP/qpoly}$ \cite{nishimura2004polynomial}. The quantum advice is given by a quantum state formed by a polynomial number of qubits. A separation between quantum and classical advice for relational problems is proved unconditionally in \cite{aaronson2023qubit}.

On another front, as machine learning gains widespread adoption, there is a growing focus on exploring the computational limits of these algorithms. In particular, this arises as a crucial condition to understand the possible applications of quantum computers in machine learning. 

Importantly, in the field of quantum machine learning, the training set, i.e., the advice that is given, does not restrict to classical pairs $(x,f(x))$. Instead, it explores various quantum generalizations of the training set. For instance, a common generalization involves a superposition of states $\ket{x}\ket{f(x)}$. In \cite{servedio2004equivalences}, an advantage over its classical counterpart is demonstrated when efficient learners are considered, i.e., learners that run in polynomial time. Furthermore, if we allow access to the circuit that generates the quantum training set, a quantum separation in sample complexity exists without the need to restrict to efficient learners \cite{salmon2023provable}. Otherwise, both the classical and quantum sample complexities are equal up to a constant \cite{arunachalam2018optimal}. Other works, as \cite{huang2021information}, consider the task of learning functions of the form $f(x)=\Tr(O \mathcal{E}(\ket{x}\bra{x}))$, where the quantum advice is given as $N$ calls to the unknown CPTP map $\mathcal{E}$. The authors show that there is no significant difference with a classical training set in the average-case. However, there is a significant performance disparity between the two in the worst-case scenario. Similar differences are also studied in \cite{gyurik2023limitations}, where the training set is of the form $\{(\rho_i^{\otimes k}),y_i\}_i$. Particularly, they demonstrate a separation between techniques that use a \textit{fully-quantum} protocol, i.e., a technique that allows for adaptive measurements during the learning process, and \textit{measure-first} protocols, shedding light on the limitations of techniques such as classical shadow tomography \cite{huang2020predicting,Koh2022classicalshadows,hu2023classical}.

While these extensions of the training set are noteworthy, quantum machine learning also places significant emphasis on scenarios involving a classical training set. Importantly, certain findings suggest the presence of quantum speed-ups in this context \cite{liu2021rigorous}. However, as illustrated in \cite{huang2021power}, achieving a (computational) quantum speed-up for a function is not a sufficient condition for a learning separation. This discrepancy arises from the fact that access to a classical training set can simplify seemingly complex tasks. Nonetheless, this mismatch disappears when considering problems for which the training set can be efficiently generated by a classical algorithm \cite{perez2023relation}.

To explore the limits of classical algorithms when provided with a training set, \cite{huang2021power} introduces the class $\mathsf{BPP/samp}$, denoting the set of languages efficiently decidable by a randomized algorithm with advice from a classical training set. The authors establish that $\mathsf{BPP}$ is a proper subset of $\mathsf{BPP/samp}$, which, in turn, is a subset of $\mathsf{P/poly}$. The first inclusion follows from undecidable unary languages, while the second inclusion is supported by an argument from Adleman's theorem. In \cite{marshall2024bounded}, other properties of the class $\mathsf{BPP/samp}$ are discussed, such as the equality $\mathsf{BPP/samp} = \mathsf{P/samp}$. Interestingly, in \cite{gyurik2023exponential}, these classes are extended to distributional problems, i.e., pairs of languages and distributions. In \cite{gyurik2023exponential}, the authors also identify several compelling examples of learning speed-ups, some of which do not rely on the classical generation of the training set.

In this work, we study the class $\mathsf{BPP/samp}$ and certain variants of it in order to understand the effects on the computational power of restricting the advice to the form of a training set. Specifically, we establish the result $\mathsf{BPP/samp}\,\subsetneq\,\mathsf{P/poly}$, demonstrating that such a restriction reduces computational power. Additionally, we explore scenarios where the distribution is fixed. In this context, we establish also a proper inclusion, indicating that sampling advice from a fixed distribution is inherently weaker than arbitrary classical advice. Furthermore, this result also holds true when we consider quantum advice and a quantum generalization of the training set. Lastly, leveraging these results and their associated proof techniques, we establish sufficient and necessary complexity assumptions for the existence of a concept class $\mathcal{C}$ that is learnable by a quantum procedure but not classically learnable. We present this result for both worst-case and average-case learning scenarios.

\section{Notation}

For the sake of completeness, we start by providing the formal definitions for the complexity classes used in our work. The classes $\mathsf{BPP/samp}$($D$), $\mathsf{BPP/samp(All)}$, $\mathsf{BQP/samp}$($D$), $\mathsf{BQP/samp(All)}$, $\mathsf{BQP/qsamp}$, and $\mathsf{BQP/qsamp}$($D$) are introduced for the first time in this work.

\begin{definition}
    $\mathsf{P/poly}$ is the class of languages decidable by a Turing Machine (TM) $M$ running in polynomial time with access to a sequence of strings $\{a_n\}_{n\in \mathbb{N}}$ of polynomial size. In other words,
    \begin{equation}
        M(x,a_{\ell(x)})=\mathbbm{1}\{x\in L\}
    \end{equation}
    for any binary string $x$, where $\ell(x)$ denotes the length of string $x$.
\end{definition}
The class $\mathsf{BQP/mpoly}$ is analogous, but it allows for a probability of error of 1/3 over the randomness of the quantum algorithm\footnote{The class $\mathsf{BQP/poly}$ differs from $\mathsf{BQP/mpoly}$ in that it includes the additional restriction that the quantum algorithm must produce a result with probability at least $\frac{2}{3}$, independently of the advice given.}. Next, we present the classes $\mathsf{BPP/samp}$($D$), $\mathsf{BPP/samp(All)}$ and $\mathsf{BPP/samp}$.
\begin{definition}
    $\mathsf{BPP/samp}$(D) is the class of languages for which there exists a probabilistic TM $M$ running in polynomial time and a polynomial $p(n)$, such that for the fixed sequence of distributions $D=\{D_n\}_{n\in \mathbb{N}}$, 
    \begin{equation}
        \mathbb{P}\left(M(x,\mathcal{T}_{p(\ell(x))})=\mathbbm{1}\{x\in L\}\right)\geq \frac{2}{3}
    \end{equation}
    holds for any fixed binary string $x$, where $\mathcal{T}_{p(n)}=\{(x_i,\mathbbm{1}\{x_i\in L\})\}_{i=1}^{p(n)}$, and $x_i\sim D_n$. The probability is taken over both the internal randomness of $M$ and the training set.
\end{definition}
\begin{definition}
    $\mathsf{BPP/samp(All)}$ is the class of languages for which there exists a probabilistic TM $M$ running in polynomial time, a sequence of distributions $\{D_n\}_{n\in \mathbb{N}}$ satisfying $\mathrm{supp}(D_n)\subseteq \{0,1\}^n$, and a polynomial $p(n)$, such that for any fixed binary string $x\in \{0,1\}^*$,
    \begin{equation}
        \mathbb{P}\left(M(x,\mathcal{T}_{p(\ell(x))})=\mathbbm{1}\{x\in L\}\right)\geq \frac{2}{3}
    \end{equation}
    where $\mathcal{T}_{p(n)}=\{(x_i,\mathbbm{1}\{x_i\in L\})\}_{i=1}^{p(n)}$ and $x_i\sim D_n$. The probability is taken over both the internal randomness of $M$ and the training set. 
\end{definition}

\begin{definition}
    $\mathsf{BPP/samp}$ is the class of languages for which there exists a probabilistic TM $M$ running in polynomial time, a sequence of distributions $\{D_n\}_{n\in \mathbb{N}}$ satisfying $\mathrm{supp}(D_n)\subseteq \{0,1\}^n$, a probabilistic polynomial time TM $G$ satisfying that the outcome of $G(1^n)$ is distributed according to $D_n$,  and a polynomial $p(n)$, such that for any fixed binary string $x\in \{0,1\}^*$,
    \begin{equation}
        \mathbb{P}\left(M(x,\mathcal{T}_{p(\ell(x))})=\mathbbm{1}\{x\in L\}\right)\geq \frac{2}{3}
    \end{equation}
    where $\mathcal{T}_{p(n)}=\{(x_i,\mathbbm{1}\{x_i\in L\})\}_{i=1}^{p(n)}$ and $x_i\sim D_n$. The probability is taken over both the internal randomness of $M$ and the training set. 
\end{definition}
Clearly, $\mathsf{BPP/samp}\,\subseteq\,\mathsf{BPP/samp(All)}$, as the definition of $\mathsf{BPP/samp}$ differs from that of $\mathsf{BPP/samp(All)}$ only by a restriction on the sequence of distributions considered. Similarly, $\mathsf{BPP/samp}$($D$)$\,\subseteq\,$$\mathsf{BPP/samp(All)}$. Regarding the relation between $\mathsf{BPP/samp}$($D$) and $\mathsf{BPP/samp}$, as we show in the next section, there exist sequences of distributions such that $\mathsf{BPP/samp}$($D$)$\,\subseteq\,$$\mathsf{BPP/samp}$, and also $\mathsf{BPP/samp}$($D$)$\,\not\subseteq\,$$\mathsf{BPP/samp}$.

The classes $\mathsf{BQP/samp}$($D$), $\mathsf{BQP/samp(All)}$ and $\mathsf{BQP/samp}$ follow naturally by replacing the probabilistic Turing machine with a quantum algorithm running in polynomial time \footnote{For the definition of $\mathsf{BQP/samp}$ TM $G$ is not substituted by a quantum algorithm.}. Moreover, when considering quantum algorithms, the possibility of utilizing a quantum state as advice emerges, representing a more general form of advice compared to a classical string. In particular, in line with the definition of $\mathsf{P/poly}$, the quantum algorithm has access to a sequence of quantum states $\{\ket{\psi}_n\}_{n\in \mathbb{N}}$, each formed by a polynomial in $n$ number of qubits. This type of advice gives rise to the class $\mathsf{BQP/qpoly}$. 

Similarly, the classical training set can be extended to a quantum state, as shown in \cite{bshouty1995learning}. In this context, the generalization constrains the quantum states to be of the form,
    \begin{equation}
        \ket{\psi_n}=\left(\sum_{x\in \{0,1\}^n} \sqrt{\mathbb{P}_{x\sim D_n}(x)}\ket{x}\otimes |\mathbbm{1}\{x\in L\}\rangle \right)^{\otimes\, p(n)}
    \end{equation}
where $D_n$ is a probability distribution over $\{0,1\}^n$, and $p(n)$ is a polynomial. This quantum state serves as a generalization of the training set, since measuring in the computational basis the state $\sum_{x\in \{0,1\}^n} \sqrt{\mathbb{P}_{x\sim D_n}(x)}\ket{x}\otimes |\mathbbm{1}\{x\in L\}\rangle$ yields a pair $(x_i,\mathbbm{1}\{x_i\in L\})$, where $x_i\sim D_n$. Therefore, by using this quantum state, we can generate a training set. Allowing advice in this form results in the class $\mathsf{BQP/qsamp(All)}$. The classes $\mathsf{BQP/qsamp}$$(D)$ and $\mathsf{BQP/qsamp}$ are obtained by adding the corresponding restrictions on the distributions considered. That is, for $\mathsf{BQP/qsamp}$$(D)$, by fixing the distribution, and for $\mathsf{BQP/qsamp}$, by requiring the existence of an efficient classical algorithm that can sample from the sequence of distributions.

Finally, we introduce the notation and concepts related to computational learning theory employed in this paper. In particular, we examine a sequence of concept classes $\mathcal{C}=\{\mathcal{C}_n\}_{n\in \mathbb{N}}$, where each set $\mathcal{C}_n$ consists of Boolean functions $c_n^{(j)}:\{0,1\}^n\rightarrow \{0,1\}$. In this work, we consider two forms of learning: worst-case learning, where precise outcomes must be produced for all values of $x$, and average-case learning, where correct outcomes must be produced for the most likely inputs. The formal definitions are given below.

\begin{definition}
 A concept class $\mathcal{C}$ is classically (quantum) worst-case learnable if there exists an efficient classical (quantum) algorithm $\mathcal{A}$, a sequence of distributions $D=\{D_n\}_{n\in\mathbb{N}}$, and a polynomial $p(n)$, such that the algorithm, given input $x\in \{0,1\}^n$ and training set $\mathcal{T}_{p(n)}^{(j)}=\{(x_i,c_n^{(j)}(x_i))\}_{i=1}^{p(n)}$ with $x_i\sim D_n$, outputs $\mathcal{A}(x,\mathcal{T}_{p(n)}^{(j)})\in \{0,1\}$ satisfying
    \begin{equation}\label{equation_worst_case_learning}
        \mathbb{P}\left(\mathcal{A}(x,\mathcal{T}_{p(n)}^{(j)})=c_n^{(j)}(x)\right)\geq 2/3
    \end{equation}
for all $x\in\{0,1\}^n$, $c_n^{(j)} \in \mathcal{C}_n $, and $n\in \mathbb{N}$. The probability is taken with respect to the training set $\mathcal{T}_{p(n)}^{(j)}$ and the randomness of algorithm $\mathcal{A}$. 
\end{definition}

\begin{definition}\label{average_case}
    A concept class $\mathcal{C}$ is classically (quantum) average-case learnable for a sequence of distributions $D$ if there exists an efficient classical (quantum) algorithm $\mathcal{A}$, and a multivariate polynomial $p(a,b)$, such that the algorithm $\mathcal{A}$, given input $x\in \{0,1\}^n$, string $1^{m}$, and training set $\mathcal{T}_{p(n,m)}^{(j)}=\{(x_i,c_n^{(j)}(x_i))\}_{i=1}^{p(n,m)}$ with $x_i\sim D_n$, outputs $\mathcal{A}(x,\mathcal{T}_{p(n,m)}^{(j)},1^{m})\in \{0,1\}$ satisfying
        \begin{equation}\label{condition_definition_1}
            \mathbb{E}_{x\sim D_n,\mathcal{T}_{p(n,m)}^{(j)},\textnormal{ }\mathcal{A}}\left[ \mathbbm{1} \left\{\mathcal{A}(x,\mathcal{T}_{p(n,m)}^{(j)},1^{m})=c_n^{(j)}(x)\right\} \right]\geq 1-\frac{1}{m}
        \end{equation}
    for all $n,m\in \mathbb{N}$, and $c_n^{(j)} \in \mathcal{C}_n $. The expected value is taken with respect to $x\sim D_n$, the training set $\mathcal{T}_{p(n,m)}^{(j)}$, and the randomness of the algorithm $\mathcal{A}$. 
\end{definition}


The algorithm $\mathcal{A}$ can be understood as the concatenation of a learning algorithm that outputs some hypothesis $h_n:\{0,1\}^n \rightarrow \{0,1\}$ and an algorithm that evaluates it at input $x$. Regarding Definition \ref{average_case}, as we show in Appendix \ref{appendix_equivalence_of_learning_def}, it can be re-expressed by changing condition \eqref{condition_definition_1} to 
\begin{equation}\label{auxiliary_form}
    \mathbb{P}_{x\sim D_n}\left( \mathbb{P}_{\mathcal{T}_{p(n,m)}^{(j)},\textnormal{ }\mathcal{A}}\left(\mathcal{A}(x,\mathcal{T}_{p(n,m)}^{(j)},1^{m})=c_n^{(j)}(x)\right)\geq \frac{2}{3}\right)\geq 1-\frac{1}{m}
\end{equation}
where the external probability is taken with respect to $x\sim D_n$, and the internal probability with respect to the training set $\mathcal{T}_{p(n,m)}^{(j)}$, and the randomness of algorithm $\mathcal{A}$. With this in mind, we see that Definition \ref{average_case} coincides with the definition of a PAC-learnable concept class—using a hypothesis class that can be evaluated efficiently—as given in \cite{gyurik2023exponential}. Finally, the 2/3 appearing in both \eqref{equation_worst_case_learning} and \eqref{auxiliary_form} is not essential. That is, it can be replaced by any constant strictly greater than 1/2, or by imposing the constraint that, given $\delta$, the probability must be greater than $1 - \delta$ \footnote{In this case, the training set size must depend on $1/\delta$.} without changing the definitions. This follows from using a majority vote. We choose this formulation for simplicity of notation.

\section{Sampling advice and Classical advice}\label{section_3}

In this section, we present the results that relate sampling advice with classical advice. We begin with the simplest case, that is, the relation between $\mathsf{BPP/samp(All)}$ and $\mathsf{P/poly}$.

\begin{theorem}\label{theorem_1}
    $\mathsf{P/poly}$ $\,\subseteq\,$ $\mathsf{BPP/samp(All)}$, which implies that 
    \begin{equation}
        \mathsf{P/poly} \,=\, \mathsf{BPP/samp(All)}.
    \end{equation}
\end{theorem}

\begin{proof}
    If $L\in \,\mathsf{P/poly}$, then there exists a sequence of strings $\{a_n\}_{n\in \mathbb{N}}$ with $\ell(a_n) = q(n)$, where $q(n)$ denotes a polynomial, and a TM $M$ running in polynomial time such that $M(x,a_{\ell(x)})=\mathbbm{1}\{x\in L\}$. The idea of the proof is to encode the string $a_n$ into a distribution $D_n$ such that, with a polynomial number of samples, we can decode with high probability $a_n$.

    For convenience, we define the distribution over the set of integers $\{0, \cdots, 2^n-1\}$ instead of over the set $\{0,1\}^n$. The encoding is quite straightforward. At position $x=0$ of the distribution, we encode the first bit of $a_n$ using the following rule: $\mathbb{P}_{x\sim D_n}(0)=p_0$ if the first bit is zero and $\mathbb{P}_{x\sim D_n}(0)=2p_0$ if it is one. At position $x=1$, we encode the second bit, and so on. This means,
    \begin{equation}
        \mathbb{P}_{x\sim D_n}(x)= \left\{\begin{matrix}
0\text{ if }x\geq \ell(a_n) \\ 
(a_n(x)+1)\,p_0 \text{ if }x< \ell(a_n) 
\end{matrix}\right.
    \end{equation}
    where $a_n(x)$ denotes the $x^{th}$ bit of string $a_n$. The constant $p_0$ can be determined using the relation:
    \begin{equation}
        1=\sum_{x=0}^{2^n-1} \mathbb{P}_{x\sim D_n}(x) = p_0 \left( q(n)+w(a_n)\right) 
    \end{equation}
    where $w(a_n)$ denotes the Hamming weight of string $a_n$. Consequently,
    \begin{equation}\label{bound_p_0}
        \frac{1}{2q(n)}\leq p_0\leq \frac{1}{q(n)}
    \end{equation}
    Next, for decoding the advice efficiently, we apply the following criterion:
    \begin{equation}
        \hat{a}_n(z)=\argmin_{a\in\{0,1\}} |aM+(1-a)m-\sum_{i=1}^N \mathbbm{1}\{x_i=z\} |
    \end{equation}
    where $x_1,\cdots,x_N$ denote i.i.d. realizations of distribution $D_n$, $M=\max_{z}\{\sum_{i=1}^N \mathbbm{1}\{x_i=z\}\}$, and $m=\min_{z}\{\sum_{i=1}^N\mathbbm{1}\{x_i=z\}:\sum_{i=1}^N \mathbbm{1}\{ x_i=z\}>0\}$.
    Using this criterion, if the error in the empirical probabilities, i.e., $\left|\mathbb{P}_{x\sim D_n}(z)-\frac{1}{N}\sum_{i=1}^N \mathbbm{1}\{x_i=z\}\right|$, is lower than, for instance, $p_0/10$ for all $z\in [0:q(n)-1]$, then the decoded string is correct.

    \smallskip
    By applying Hoeffding's inequality, we have that for $N=\frac{1}{2\epsilon^2} \log\frac{2}{\delta}$,
    \begin{equation}
        \mathbb{P}\left(\left| \frac{1}{N}\sum_{i=1}^N{\mathbbm{1}\{x_i=z\}} -\mathbb{P}_{x\sim D_n}(z)  \right|\leq \epsilon\right)\geq 1-\delta
    \end{equation}
     Therefore, using the union bound, and     taking $\epsilon= p_0/10$, and $\delta=1/(3q(n))$, it follows that 
    \begin{equation}
        \mathbb{P}\left( \bigcup_{z\in [0:q(n)-1]} \left\{ \left| \frac{1}{N}\sum_{i=1}^N{\mathbbm{1}\{x_i=z\}} - \mathbb{P}_{x\sim D_n}(z) \right|> \frac{p_0}{10} \right\}\right)< \frac{1}{3}
    \end{equation}
    if the number of samples in the training set is 
    \begin{align}
        N =\frac{50}{p_0^2} \log\left( 6 \,q(n)\, \right) &\leq 200 \, q(n)^2 \log\left( 6 \,q(n)\, \right) \nonumber \\&\leq 1200q(n)^3 
    \end{align}
    where the first inequality follows from \eqref{bound_p_0}, and the second inequality uses $\log x \leq x $.
    
    That is, the probability of incorrectly decoding the string $a_n$ is at most $1/3$ when $N=1200\,q(n)^3$. Therefore, by taking a polynomial number of samples, we can generate an estimate of the advice $\hat{a}_n$ such that $\mathbb{P}(\hat{a}_n=a_n)\geq 2/3$. Consequently, there exists a TM $M'$, a distribution, and a polynomial $p(n)$ such that
    \begin{equation}
        \mathbb{P}\left(M'(x,\mathcal{T}_{p(\ell(x))})=\mathbbm{1}\{x\in L\}\right)\geq \frac{2}{3}
    \end{equation}
    That is, $L\in\,\mathsf{BPP/samp(All)}$.

\end{proof}

Analogously, by the same proof, $\mathsf{BQP/mpoly}$ $\subseteq$ $\mathsf{BQP/samp(All)}$, implying the equality $\mathsf{BQP/mpoly}= \mathsf{BQP/samp(All)}$. Interestingly, the equality between these classes arises not because the pairs $(x_i,\mathbbm{1}\{x_i\in L\})$ contained in the training set are inherently powerful, but rather because we can encode the advice in the distribution used for $x$. One rapidly realizes that the same trick does not work if the distribution is fixed. In that setting, one might wonder if there is some distribution such that $\mathsf{BPP/samp}$($D$) equals $\mathsf{P/poly}$, which is equivalent to asking if classical advice is equivalent to sampling advice distributed according to $D$. The subsequent result provides a negative answer to this question.
\begin{theorem}\label{theorem_distribution_fixed}
    For any sequence of distributions $D$, $\mathsf{BPP/samp}$\textnormal{($D$)}$\,\subsetneq$ $\mathsf{P/poly}$.
\end{theorem}
The proof of this theorem follows trivially from the combination of Lemma \ref{lemma_1} and Lemma \ref{lemma_2}, which we now state. 
\begin{lemma}\label{lemma_1}
    If there exists a concept class $\mathcal{C}$ that satisfies 
    \noindent

    \begin{enumerate}[$(i)$]
        \item $\mathcal{C}$ is not classically worst-case learnable. 
        
        \item All Boolean functions $c_n^{(j)}$ can be efficiently computed by a classical algorithm, i.e., there exists an efficient classical algorithm such that given $j$ and $x\in \{0,1\}^n$ outputs $c_n^{(j)}(x)$ in polynomial time.
    
        \item $|\mathcal{C}_n|\leq 2^{q(n)}$ for all $n\in\mathbb{N}$, where $q(n)$ is a polynomial. 
    \end{enumerate}

    \smallskip\smallskip\noindent
    then $\mathsf{BPP/samp}\textnormal{($D$)}\,\subsetneq\,$$\mathsf{P/poly}$ for any sequence of distributions $D=\{D_n\}_{n\in \mathbb{N}}$.
\end{lemma}

\begin{proof} The idea of this proof is to construct from the concept class $\mathcal{C}$ a language $L(\mathcal{C})$ such that $L(\mathcal{C}) \notin \,$$\mathsf{BPP/samp}$($D$) but $L(\mathcal{C}) \in \,$$\mathsf{P/poly}$. In particular, $L(\mathcal{C}):=\{x:c_{\ell(x)}^{\left(g(\ell(x))\right)}(x)=1\}$ where function $g(n)\in \{1,\cdots, |\mathcal{C}_n|\}$. However, before defining $g(n)$, we need to introduce some extra notation. Given an algorithm $\mathcal{A}$, a sequence of distributions $D$, and a polynomial $p(n)$, we define the set $\mathcal{E}_n(\mathcal{A},D,p)\subseteq \mathcal{C}_n$, as the set of concepts $c_n^{(j)}$  that satisfy
    \begin{equation}
        \mathbb{P}\left(\mathcal{A}(x,\mathcal{T}_{p(n)}^{(j)})=c_n^{(j)}(x)\right)< 2/3
    \end{equation}
for some $x\in\{0,1\}^n$. Therefore, if $\mathcal{C}$ is not worst-case learnable, then any triplet $(\mathcal{A},D,p)$ satisfies that $\mathcal{E}_n(\mathcal{A},D,p)\neq\emptyset$ for some values $n\in \mathbb{N}$. The set of values of $n\in \mathbb{N}$ for which this occurs is denoted by
\begin{equation}
    \mathcal{N}(\mathcal{A},D,p):=\{n\in \mathbb{N}: \mathcal{E}_n(\mathcal{A},D,p)\neq \emptyset\}
\end{equation}
Additionally, we use $B_{\mathrm{alg}}$ to represent a bijection between $\mathbb{N}$ and the set of efficient classical algorithms, and $B_{\mathrm{poly}}$ to denote a bijection between $\mathbb{N}$ and the set of polynomials $\{c\,n^k:(c,k)\in\mathbb{N}^2\}$.

The motivation behind the definition of function $g(n)$ is to ensure that for any pair $(\mathcal{A},p)$, there exists some $n$ such that $c_n^{(g(n))}$ is not learned in the worst-case by the algorithm $\mathcal{A}$ using a training set of size $p(n)$. Now, we provide the definition,
\begin{equation}
    g(n):=\left\{\begin{matrix}
1 \text{ if }n\notin\{n_i\}_{i=1}^{\infty}\\ 
\min \{j: c_n^{(j)}\in \mathcal{E}_n\left(B_{\mathrm{alg}}(s_{i^*(n)}),D,B_{\mathrm{poly}}(z_{i^*(n)})\right)\}\text{ if }n\in\{n_i\}_{i=1}^{\infty}
\end{matrix}\right.
\end{equation}
    where sequences $s_i\in \mathbb{N}$ and $z_i\in \mathbb{N}$ satisfy that any point $(a,b)\in \mathbb{N}^2$ appears at least once in the sequence $(s_i,z_i)$, $n_1:=\min\{\mathcal{N}\left(B_{\mathrm{alg}}(s_{1}),D,B_{\mathrm{poly}}(z_{1})\right) \}$, and $n_i=\min\{\mathcal{N}\left(B_{\mathrm{alg}}(s_{i}),D,B_{\mathrm{poly}}(z_{i})\right) \backslash\{n_1,n_2,\cdots,n_{i-1}\} \}$. Finally, $i^*(n)$ denotes the index $i$ such that $n_i=n$. Note that the set 
    \begin{equation}
        \mathcal{N}\left(B_{\mathrm{alg}}(s_{i}),D,B_{\mathrm{poly}}(z_{i})\right) \backslash\{n_1,n_2,\cdots,n_{i-1}\}
    \end{equation}
    cannot be the empty set from condition $(i)$, which implies that the set $\mathcal{N}(\mathcal{A},D,p)$ is formed by a countable infinite number of elements for any pair $(\mathcal{A},p)$.

    Next, we proceed to prove that $L(\mathcal{C}) \notin \,$$\mathsf{BPP/samp}$($D$), but $L(\mathcal{C}) \in \,\mathsf{P/poly}$. It is straightforward to verify that $L(\mathcal{C})\in\,\mathsf{P/poly}$. Given $g(n)$, which can be represented with a polynomial in $n$ number of bits (condition $(iii)$), we can efficiently compute $c_n^{(g(n))}(x)$ using a classical algorithm, as the concepts $c_n^{(j)}$ can be computed efficiently (condition $(ii)$). To establish that $L(\mathcal{C})\notin\,\mathsf{BPP/samp}$($D$), we employ a contradiction argument. In other words, we assume that $L(\mathcal{C})\in \,\mathsf{BPP/samp}$($D$), indicating the existence of a pair $(\mathcal{A},p)$ such that
    \begin{equation}
        \mathbb{P}\left(\mathcal{A}(x,\mathcal{T}_{p(n)})=\mathbbm{1}\{x\in L(\mathcal{C})\}\right)\geq 2/3
    \end{equation}
    for all $x\in\{0,1\}^n$, and $n\in \mathbb{N}$. Or equivalently, 
    \begin{equation}
        \mathbb{P}\left(\mathcal{A}(x,\mathcal{T}_{p(n)}^{(g(n))})=c_n^{(g(n))}(x)\right)\geq 2/3
    \end{equation}
    Since each point $(a,b)\in \mathbb{N}^2$ appears at least once in the sequence $(s_i,z_i)$, there exists a value $i$ such that $(s_i,z_i)=(B_{\mathrm{alg}}^{-1}(\mathcal{A}),B_{\mathrm{poly}}^{-1}(p))$. Therefore, for $n=n_i$, from the definition of $g(n)$, we have that 
    \begin{equation}
        \mathbb{P}\left(\mathcal{A}(x,\mathcal{T}_{p(n_i)}^{(g(n_i))})=c_{n_i}^{(g(n_i))}(x)\right)< 2/3
    \end{equation}
    for some $x\in \{0,1\}^{n_i}$. Therefore, there is a contradiction, and consequently, $L(\mathcal{C})\notin\,\mathsf{BPP/samp}$($D$).

\end{proof}

\begin{lemma}\label{lemma_2}
    The concept class $\mathcal{C}_n=\{c_n^{(j)}(x)=\mathbbm{1}\{x=j-1\}\}_{j=1}^{2^n}$ satisfies:

    \begin{enumerate}[$(i)$]
        \item It is not classically worst-case learnable. 
        
        \item All Boolean functions $c_n^{(j)}$ can be efficiently computed by a classical algorithm.
    
        \item $|\mathcal{C}_n|= 2^{n}$ for all $n\in\mathbb{N}$.
    \end{enumerate}
\end{lemma}
\begin{proof}
See Appendix \ref{appendix_A}.
\end{proof}

From the proof of Lemma \ref{lemma_1} and using Theorem \ref{theorem_1}, another interesting property of the class $\mathsf{BPP/samp}$($D$) becomes apparent. There is no optimal sequence of distributions $D'$, i.e., a distribution that satisfies $\mathsf{BPP/samp}$($D$)$\,\subseteq\,\mathsf{BPP/samp}$($D'$) for all $D$. Moreover, this can be extended to the following proposition.

\begin{proposition}
    There does not exist a countably infinite set of sequences $\{D^{i}\}_{i\in \mathbb{N}}$ such that  $\mathsf{BPP/samp}$\textnormal{($D$)}$\subseteq \bigcup_i\mathsf{BPP/samp}$\textnormal{($D^{i}$)} holds for all $D$.
\end{proposition}

\begin{proof}
    Let us assume that there exists a countably infinite set of sequences $\{D^{i}\}_{i \in \mathbb{N}}$ such that $\mathsf{BPP/samp}(D) \subseteq \bigcup_i \mathsf{BPP/samp}$($D^{i}$) holds for all $D$. This implies that $\bigcup_i \mathsf{BPP/samp}(D^{i}) = \mathsf{BPP/samp(All)} = \mathsf{P/poly}$. On the other hand, note that we can modify the language $L(\mathcal{C})$, as defined in Lemma \ref{lemma_1}, by changing function $g(n)$ to
    \begin{equation}
        g(n):=\left\{\begin{matrix}
    1 \text{ if }n\notin\{n_i\}_{i=1}^{\infty}\\ 
    \min \{j: c_n^{(j)}\in \mathcal{E}_n\left(B_{\mathrm{alg}}(s_{i^*(n)}),B_{\mathrm{dist}}(v_{i^*(n)}),B_{\mathrm{poly}}(z_{i^*(n)})\right)\}\text{ if }n\in\{n_i\}_{i=1}^{\infty}
    \end{matrix}\right.
    \end{equation}
    where $B_{\mathrm{dist}}$ is defined as $B_{\mathrm{dist}}(k):=D^k$, the sequence $(s_i,v_i,z_i)$ satisfies that every point in $\mathbb{N}^3$ appears at least once, $n_1:=\min\{\mathcal{N}\left(B_{\mathrm{alg}}(s_{1}),B_{\mathrm{dist}}(v_1),B_{\mathrm{poly}}(z_{1})\right) \}$, and $n_i=\min\{\mathcal{N}\left(B_{\mathrm{alg}}(s_{i}),B_{\mathrm{dist}}(v_{i}),B_{\mathrm{poly}}(z_{i})\right) \backslash\{n_1,n_2,\cdots,n_{i-1}\} \}$. Using the same proof as in Lemma \ref{lemma_1} and using also Lemma \ref{lemma_2}, we can conclude that $L(\mathcal{C})$ does not belong to $\bigcup_i$$\mathsf{BPP/samp}$($D^{i}$), but it belongs to $\mathsf{P/poly}$. This yields a contradiction, and therefore, there does not exist a countably infinite set of sequences $\{D^{i}\}_{i\in \mathbb{N}}$ such that $\mathsf{BPP/samp}$($D$)$\subseteq \bigcup_i \mathsf{BPP/samp}$($D^{i}$) holds for all $D$.

\end{proof}

Note that since the set of efficient probabilistic Turing machines is countable, $\mathsf{BPP/samp}$ can be expressed as $\bigcup_{i\in \mathbb{N}} \mathsf{BPP/samp}$ (${D}^{i}$), where the set of sequences $\{{D}^{i}\}_{i\in \mathbb{N}}$ consists of those for which there exists an efficient sampling algorithm. Therefore, the next result follows:
\begin{corollary}
    $\mathsf{BPP/samp}\,\subsetneq\,\mathsf{P/poly}$.
\end{corollary}
\begin{proof}
    Assume that $\mathsf{BPP/samp}\,=\,\mathsf{P/poly}$. Then, there exists a countably infinite set of sequences $\{{D}^{i}\}_{i\in \mathbb{N}}$ such that 
    \begin{equation}
    \bigcup_{i\in \mathbb{N}} \mathsf{BPP/samp}({D}^{i})\,=\, \mathsf{P/poly}
    \end{equation}
    Therefore, since $\mathsf{P/poly}\,=\, \mathsf{BPP/samp(All)}$, it follows that there exists a countably infinite set of sequences $\{{D}^{i}\}_{i\in \mathbb{N}}$ such that  
    \begin{equation}
        \mathsf{BPP/samp}({D})\subseteq \bigcup_{i\in \mathbb{N}}\mathsf{BPP/samp}({D}^{i})
    \end{equation}
    holds for all ${D}$, which is a contradiction.
\end{proof}

Hence, this resolves the question raised in \cite{huang2021power} regarding the comparison between $\mathsf{BPP/samp}$ and $\mathsf{P/poly}$. Furthermore, this, together with the result $\mathsf{P/poly} \,=\, \mathsf{BPP/samp(All)}$, implies that there exists some sequence of distributions such that $\mathsf{BPP/samp}(D)\not\subseteq\,\mathsf{BPP/samp}$. Finally, to conclude this section, we consider a curiosity that stems from the main point of the proof of Lemma \ref{lemma_1}, i.e., the construction from a concept class of a language that demonstrates a separation between the classes $\mathsf{BPP/samp}$($D$) and $\mathsf{P/poly}$. It may be interesting to contemplate the opposite: Can we generate non-trivial concept classes that exhibit similar properties to the ones used in Lemma \ref{lemma_1} from languages $L \,\in \,\mathsf{P/poly}\backslash\mathsf{BPP/samp}$($D$)?. The following result partially answers this question. 
\begin{proposition}
    For each language $L\in \, \mathsf{P/poly}\,\backslash \mathsf{BPP/samp}$\textnormal{($D$)}, there exists a concept class $\mathcal{C}(L)$ such that 
    \begin{enumerate}[$(i)$]
        \item $\mathcal{C}(L)$ is not classically worst-case learnable using the sequence of distributions $D$.

        \item All Boolean functions $c_n^{(j)}$ can be efficiently computed by a classical algorithm.
    
        \item $2^{n}\leq|\mathcal{C}_n(L)|\leq  2^{p(n)}$ for all $n\in\mathbb{N}$, where $p(n)$ is a polynomial.
    \end{enumerate}

\begin{proof}
    Since $L\in \, \mathsf{P/poly}$, we have that there exists a sequence of strings $\{a_n\}_{n\in \mathbb{N}}$ such that $\ell(a_n)= q(n)$, where $q(n)$ denotes a polynomial, and a TM $M$ running in polynomial time such that $M(x,a_{\ell(x)})=\mathbbm{1}\{x\in L\}$. Therefore, we can express the language $L$ as the set $\{x : M(x, a_{\ell(x)}) = 1\}$. Next, we define our learning problem using the Turing machine  $M$ as $ c_n^{(a)}(x) = M(x, a)$, where $\mathcal{C}_n(L) := \{M(x, a) : \ell(a) = q(n) \text{ and } \ell(x) = n\}$. However, note that this definition only implies that $ |\mathcal{C}_n(L)| \leq 2^{q(n)}$, not that $ |\mathcal{C}_n(L)| \geq 2^n$. This is because it may happen that two different strings $i$ and $j$ satisfy $M(x, i) = M(x, j)$ for all $x \in \{0,1\}^n$. To ensure that sufficiently many distinct concepts are defined, we introduce an auxiliary Turing machine $M'$:
    \begin{equation}
        M'(x,a)= \left\{\begin{matrix}
        M(x,a) & \text{if}\; \ell(a)\leq n \\ M(x,a_{[n+1:\ell(a)]}) & \text{if}\; \ell(a)>n\text{ and } a_{[1:n]}= 0^n \\  
        M(x,a_{[n+1:\ell(a)]}) \oplus \mathbbm{1} \{x=a_{[1:n]}\} & \text{if}\;  \ell(a)>n\text{ and } a_{[1:n]} \neq 0^n
        \end{matrix}\right.
    \end{equation}
    where $a_{[u:v]}$ denotes the substring of $a$ from bits $u$ to $v$.

    Using this auxiliary TM $M'$, we define the concept class as $\mathcal{C}_n(L):=\{M'(x,a):\ell(a)=q(n)+n \text{ and } n=\ell(x)\}$. Clearly, this new concept class satisfies $(a)$ $|\mathcal{C}_n(L)|\geq  2^{n}$, $(b)$ the different concepts $c_n^{(j)}(x)$ can be computed efficiently since TM $M$ runs in polynomial time, and $(c)$ $L=\{x: c_n^{(g(n))}(x)=1\}$ for $g(n)=0^n\#a_n$, where $\#$ denotes the concatenation operation. Hence, $\mathcal{C}_n(L)$ satisfies conditions ($ii$) and ($iii$).

    Finally, to prove ($i$), we use a contradiction argument. That is, we assume that $\mathcal{C}(L)$ is classically worst-case learnable using sequence $D$. Therefore, there exists an algorithm $\mathcal{A}$ and polynomial $p(n)$ such that 

    \begin{equation}
        \mathbb{P}\left( \mathcal{A}(x,\mathcal{T}_{p(n)}^{(j)})=c_n^{(j)}(x)\right)\geq 2/3
    \end{equation}
    for all $x \in \{0,1\}^n$, $c_n^{(j)}\in \mathcal{C}_n(L)$, and $n\in \mathbb{N}$. However, since $L\notin\,\mathsf{BPP/samp}$($D$), for any efficient classical algorithm $\mathcal{A}$ and polynomial $p(n)$, there exists a non-empty set $\mathcal{N}\subseteq \mathbb{N}$ such that for all $n\in \mathcal{N}$
    \begin{equation}
        \mathbb{P}\left( \mathcal{A}(x,\mathcal{T}_{p(n)}^{(g(n))})=c_n^{(g(n))}(x)\right)< 2/3
    \end{equation}
    for some $x\in \{0,1\}^n$. Therefore, there is a contradiction. Consequently, $\mathcal{C}(L)$ is not classically worst-case learnable using sequence $D$.
    
 \end{proof}
\end{proposition}
\section{Quantum sampling advice and Quantum advice}

In this section, we explore a complete quantum setting, meaning that we not only consider quantum algorithms but also incorporate quantum advice and the quantum generalization of the training set. Interestingly, similar ideas to those presented in Theorem \ref{theorem_1} do not generalize to this setting. That is, we cannot discern whether $\mathsf{BQP/qsamp(All)}$ is a proper subset of $\mathsf{BQP/qpoly}$, or if both classes are equal. However, the results for a fixed distribution do generalize to the quantum setting. This is formally stated below.

\begin{theorem}\label{quantum_sampling}
    For any sequence of distributions $D$, $\mathsf{BQP/qsamp}$\textnormal{($D$)} is a proper subset of $\mathsf{BQP/qsamp(All)}$.
\end{theorem}

Similarly to Theorem \ref{theorem_distribution_fixed}, this theorem follows from two lemmas.

\begin{lemma}\label{lemma_3}
    If there exists a concept class $\mathcal{C}$ that satisfies 
    \noindent

    \begin{enumerate}[$(i)$]
        \item $\mathcal{C}$ is not worst-case learnable by a quantum algorithm given quantum sampling advice.

        \item All Boolean functions $c_n^{(j)}$ can be efficiently computed by a classical algorithm.
    
        \item $|\mathcal{C}_n|\leq 2^{q(n)}$ for all $n\in\mathbb{N}$, where $q(n)$ is a polynomial. 
    \end{enumerate}

    \smallskip\smallskip\noindent
    then $\mathsf{P/poly} \,\not\subseteq \, \mathsf{BQP/qsamp}$\textnormal{($D$)} for any sequence of distributions $D=\{D_n\}_{n\in \mathbb{N}}$.
\end{lemma}

The proof of this lemma is completely analogous to that of Lemma \ref{lemma_2}. That is, we rely on the same language $L(\mathcal{C})$, where the definition of function $g(n)$ is analogous. However, instead of considering classical algorithms, we employ quantum algorithms.

\begin{lemma}\label{lemma_4}
    The concept class $\mathcal{C}_n=\{c_n^{(j)}(x)=\mathbbm{1}\{x=j-1\}\}_{j=1}^{2^n}$ satisfies:

    \begin{enumerate}[$(i)$]
        \item $\mathcal{C}$ is not worst-case learnable by a quantum algorithm given quantum sampling advice.

        \item All Boolean functions $c_n^{(j)}$ can be efficiently computed by a classical algorithm.
    
        \item $|\mathcal{C}_n|= 2^{n}$ for all $n\in\mathbb{N}$.
    \end{enumerate}
\end{lemma}
\begin{proof}
 See Appendix \ref{proof_lemma4}.  
\end{proof}

Therefore, by combining both lemmas, we conclude that $\mathsf{P/poly}\,\not\subseteq \, \mathsf{BQP/qsamp}$$(D)$. Consequently, since $\mathsf{P/poly}\,\subseteq \,\mathsf{BQP/mpoly} \,= \, \mathsf{BQP/samp(All)}\,\subseteq \, \mathsf{BQP/qsamp(All)}$, the class $\mathsf{BQP/qsamp}$$(D)$ must be a proper subset of $\mathsf{BQP/qsamp(All)}$. Otherwise, if $\mathsf{BQP/qsamp}$$(D)$ were equal to $\mathsf{BQP/qsamp(All)}$, we would arrive at a contradiction. As in the previous section, this result can be generalized to a countably infinite set of sequences $\{D^i\}_{i \in \mathbb{N}}$, implying that $\mathsf{BQP/qsamp}$ is a proper subset of $\mathsf{BQP/qsamp(All)}$.

As a summary of the results from this section and Section \ref{section_3}, the figure below illustrates the various relations between the studied complexity classes.

 \begin{figure}[H]
	\centering
	\includegraphics[width=16cm]{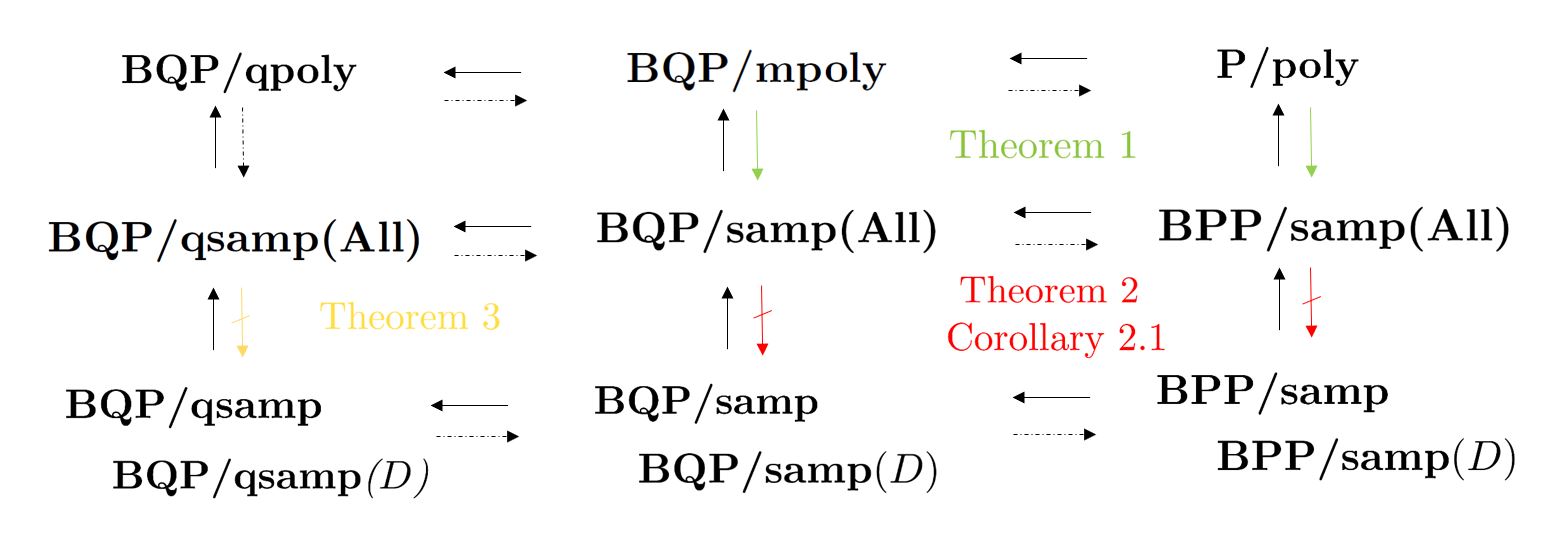}
	\caption{Diagram of the relations between the different complexity classes studied. A solid arrow from $A$ to $B$ implies $A \subseteq B$. If the arrow is dashed, then it is unknown if that inclusion holds. If a bar is included, the inclusion does not hold. The last row, for simplicity, contains two classes. The horizontal arrows relate classes of the same type—i.e., either $\mathsf{/samp}$ or $\mathsf{/samp}(D)$—while the vertical arrows apply to both.}
	\label{fig:generate_samples}
\end{figure}

\section{Complexity assumptions for learning speed-ups}

Up to this point, our focus has primarily centered on the vertical relations illustrated in Figure 1. However, the study of the other relations is also highly interesting. In this section, we show how some of these relations are intricately tied to the existence of learning separations. We examine three main scenarios: two concerning the worst-case setting and one involving the average-case setting. Specifically, in the two worst-case scenarios, we examine both the case where the separation is tailored to a specific distribution and the case where, regardless of the distribution used, the concept class is not classically learnable. For the latter scenario, the following result holds.

\begin{theorem}\label{theorem_worst_case}
    The condition $\mathsf{BQP}\, \not\subset \, \mathsf{P/poly}$ is sufficient for the existence of a concept class $\mathcal{C}$ that is worst-case learnable by a quantum procedure but is not classically learnable. Additionally, the condition $\mathsf{PromiseBQP} \, \not\subset \, \mathsf{PromiseP/poly}$ is a necessary condition for the existence of a concept class exhibiting the aforementioned property.
\end{theorem}

\begin{proof}
    See Appendix \ref{proof_theorem3}.
\end{proof}
Before proceeding with the discussion of this result, let us first clarify the statement of the theorem. The classes $\mathsf{PromiseBQP}$ and $\mathsf{PromiseP/poly}$ denote the analogous classes to $\mathsf{BQP}$ and $\mathsf{P/poly}$ for \textit{promise problems}. A promise problem is a generalization of a decision problem where the decision is only defined on a subset $S$ of $\{0,1\}^*$. Consequently, an algorithm is considered to solve a promise problem if it produces correct answers for the strings in $S$, while it has the freedom to output any result for the remaining strings.

Now that the statement is clear, we observe that the necessary and sufficient conditions are quite similar, with only a subtle difference between them. Regarding the necessary condition $\mathsf{PromiseBQP} \, \not\subset \, \mathsf{PromiseP/poly}$, it implies that quantum computations cannot be efficiently simulated classically, even when augmented with polynomial-size classical advice. Thus, this condition is stronger than the more standard condition $\mathsf{PromiseBQP}\, \not\subset \, \mathsf{PromiseBPP}$. Next, in the less restricted scenario where the speed-up is considered only for a particular distribution, the required conditions become weaker, as shown below.

\begin{theorem}\label{fixed_distribution_thm}
    The condition $\mathsf{BPP/samp}(D)\, \subsetneq \, \mathsf{BQP/samp}$\textnormal{($D$)} is a necessary and sufficient condition for the existence of a concept class $\mathcal{C}$ that is worst-case learnable by a quantum procedure using sequence $D$ but it is not worst-case classically learnable using sequence $D$. 
\end{theorem}

\begin{proof}
    Firstly, we establish the sufficiency of the statement. Therefore, we assume $\mathsf{BPP/samp}(D)\, \subsetneq \, \mathsf{BQP/samp}$($D$), implying the existence of a language $L$ such that $L\in \mathsf{BQP/samp}$($D$) and $L\notin \mathsf{BPP/samp}$($D$). Consequently, a concept class formed solely by the concept $c_n(x)=\mathbbm{1}\{x \in L\}$ is quantum learnable using sequence $D$ by definition, and it is not classically learnable using sequence $D$.
    
    Regarding the necessary part, we assume the existence of a concept class $\mathcal{C}$ that satisfies the conditions stated in the theorem. Next, using the construction of the language $L(\mathcal{C})$ employed in Lemma \ref{lemma_1}, we establish that $L(\mathcal{C}) \notin \,\mathsf{BPP/samp}$($D$). Since the concept class is (worst-case) learnable by a quantum procedure, we conclude that $L(\mathcal{C}) \in \,\mathsf{BQP/samp}$($D$). Consequently, $\mathsf{BPP/samp}$($D$) is a proper subset of $\mathsf{BQP/samp}$($D$).
\end{proof}

Interestingly, we can also show that $\mathsf{BQP}\,\not\subset\,\mathsf{BPP/samp}$($D$) is a sufficient condition, while $\mathsf{PromiseBQP}\,\not\subset\,\mathsf{PromiseBPP/samp}$($D$) is a necessary one. Note that this is analogous to the result of Theorem \ref{theorem_worst_case}. Consequently, by analogy, Theorem \ref{fixed_distribution_thm} suggests that the true necessary and sufficient condition for the scenario outlined in Theorem \ref{theorem_worst_case} might be $\mathsf{P/poly}\,\subsetneq \,\mathsf{BQP/mpoly}$. Finally, we examine quantum speed-ups for the average-case scenario.

\begin{theorem}\label{theorem_average_case}
    The condition $\mathsf{HeurBPP/samp}(D)\,\subsetneq\,\mathsf{HeurBQP/samp}$$(D)$ is a necessary and sufficient condition for the existence of a concept class $\mathcal{C}$ that is average-case learnable by a quantum procedure for sequence $D$ but it is not classically average-case learnable for sequence $D$.
\end{theorem}
\begin{proof}
    See Appendix \ref{proof_average}.
\end{proof}

The classes $\mathsf{HeurBPP/samp}$$(D)$ and $\mathsf{HeurBQP/samp}$$(D)$ are heuristic complexity classes, that is, sets of distributional problems $(L,D)$, where $L$ is a language and $D$ is a sequence of distributions. The specific definitions are given in Appendix \ref{proof_average}. The intuitive idea is that the algorithms are not required to work for all inputs $x$, instead, they are required to work for most inputs, where `most' is measured by the distributions $D_n$.

In summary, these theorems imply that proving learning speed-ups requires assumptions more involved than the conventional conjecture $\mathsf{BQP}\neq \mathsf{BPP}$. However, it is also worth noting that, using the discrete logarithm example provided in \cite{liu2021rigorous}, the classical hardness of the discrete logarithm implies $\mathsf{BQP}\, \not\subset \, \mathsf{BPP/samp}$($U$), where $U$ denotes the sequence of uniform distributions. Therefore, in general, these conjectures do not seem unlikely. 

\section{Conclusions}

In conclusion, we have proven that advice in the form of a training set is weaker than classical advice, i.e., $\mathsf{BPP/samp}$ is a proper subset of $\mathsf{P/poly}$. Similarly, when the distribution is fixed, the sampling advice also becomes strictly weaker than  general classical advice. These results also hold for the scenario where \textit{quantum training sets} are considered. However, unlike in the classical scenario, the relationship between quantum advice and the class $\mathsf{BQP/qsamp(All)}$ remains uncertain. Finally, we have also presented complexity assumptions that are both sufficient and necessary for the existence of a learning speed-up. In particular, we study both the worst-case and average-case scenarios. The resulting conditions, in both cases, are more involved than the usual conjecture $\mathsf{BQP} \neq \mathsf{BPP}$.

\section{Acknowledgements}
I thank Javier R. Fonollosa and Alba Pagès-Zamora for their guidance and support. This work has been funded by grants PID2022-137099NB-C41, PID2019-104958RB-C41 funded by MCIN/AEI/10.13039/501100011033 and FSE+ and by grant 2021 SGR 01033 funded by AGAUR, Dept. de Recerca i Universitats de la Generalitat de Catalunya 10.13039/501100002809.

\bibliographystyle{ieeetr}
\bibliography{bibliography.bib}

\begin{appendices}

\section{Equivalence of Learning Definitions}\label{appendix_equivalence_of_learning_def}

In this appendix, we show that Definition \ref{average_case} can be re-expressed as follows:

\begin{definition}\label{aux_definition}
    A concept class $\mathcal{C}$ is classically (quantum) average-case learnable for a sequence of distributions $D$ if there exists an efficient classical (quantum) algorithm $\mathcal{A}$, and a multivariate polynomial $p(a,b)$, such that the algorithm $\mathcal{A}$, given input $x\in \{0,1\}^n$, string $1^{m}$, and training set $\mathcal{T}_{p(n,m)}^{(j)}=\{(x_i,c_n^{(j)}(x_i))\}_{i=1}^{p(n,m)}$ with $x_i\sim D_n$, outputs $\mathcal{A}(x,\mathcal{T}_{p(n,m)}^{(j)},1^{m})\in \{0,1\}$ satisfying
    \begin{equation}
        \mathbb{P}_{x\sim D_n}\left( \mathbb{P}_{\mathcal{T}_{p(n,m)}^{(j)},\textnormal{ }\mathcal{A}}\left(\mathcal{A}(x,\mathcal{T}_{p(n,m)}^{(j)},1^{m})=c_n^{(j)}(x)\right)\geq \frac{2}{3}\right)\geq 1-\frac{1}{m}
    \end{equation}
    for all $n,m\in \mathbb{N}$, and $c_n^{(j)} \in \mathcal{C}_n $. The external probability is taken with respect to $x\sim D_n$, and the internal probability with respect to the training set $\mathcal{T}_{p(n,m)}^{(j)}$, and the randomness of algorithm $\mathcal{A}$.
\end{definition}

\begin{proof}

First, let's prove that if a concept class $\mathcal{C}$ is average-case learnable for sequence of distributions $D$ following Definition \ref{average_case}, then it is also average-case learnable according to Definition \ref{aux_definition}.
    \begin{align}\label{inequality_for_definition}
        1-\frac{1}{m} &\leq \mathbb{E}_{x\sim D_n,\mathcal{T}_{p(n,m)}^{(j)},\textnormal{ }\mathcal{A}}\left[ \mathbbm{1} \left\{\mathcal{A}(x,\mathcal{T}_{p(n,m)}^{(j)},1^{m})=c_n^{(j)}(x)\right\} \right] \nonumber \\ &= \mathbb{E}_{x\sim D_n}\left[ \mathbb{P}_{\mathcal{T}_{p(n,m)}^{(j)},\textnormal{ }\mathcal{A}} \left\{\mathcal{A}(x,\mathcal{T}_{p(n,m)}^{(j)},1^{m})=c_n^{(j)}(x)\right\} \right] \nonumber \\ & =\mathbb{E}_{x\sim D_n}\left[ \mathbb{P}_{\mathcal{T}_{p(n,m)}^{(j)},\textnormal{ }\mathcal{A}} \left\{\mathcal{A}(x,\mathcal{T}_{p(n,m)}^{(j)},1^{m})=c_n^{(j)}(x)\right\} |x\in \mathcal{X}_{2/3}(n,m)\right] \mathbb{P}_{x\sim D_n} (x \in \mathcal{X}_{2/3}(n,m)) \nonumber \\ & + \mathbb{E}_{x\sim D_n}\left[ \mathbb{P}_{\mathcal{T}_{p(n,m)}^{(j)},\textnormal{ }\mathcal{A}} \left\{\mathcal{A}(x,\mathcal{T}_{p(n,m)}^{(j)},1^{m})=c_n^{(j)}(x)\right\} |x\notin \mathcal{X}_{2/3}(n,m)\right] \mathbb{P}_{x\sim D_n} (x \notin \mathcal{X}_{2/3}(n,m)) \nonumber \\&\leq \mathbb{P}_{x\sim D_n} (x \in \mathcal{X}_{2/3}(n,m))+ \frac{2}{3}\, \mathbb{P}_{x\sim D_n} (x \notin \mathcal{X}_{2/3}(n,m)) \nonumber \\ & = \frac{2}{3}+\frac{1}{3}  \mathbb{P}_{x\sim D_n} (x \in \mathcal{X}_{2/3}(n,m))
    \end{align}
    where $\mathcal{X}_{2/3}(n,m)$ denotes the set of Boolean strings $x\in \{0,1\}^n$ such that
    \begin{equation}
        \mathbb{P}_{\mathcal{T}_{p(n,m)}^{(j)},\textnormal{ }\mathcal{A}}\left(\mathcal{A}(x,\mathcal{T}_{p(n,m)}^{(j)},1^{m})=c_n^{(j)}(x)\right)\geq \frac{2}{3}
    \end{equation}
    Inequality \eqref{inequality_for_definition} implies that $\mathbb{P}_{x\sim D_n} (x \in \mathcal{X}_{2/3}(n,m))\geq 1-\frac{3}{m}$, which proves the desired implication. Now, let's move on to the other direction. For this, we use the fact that the condition
    \begin{equation}
        \mathbb{P}_{x\sim D_n}\left( \mathbb{P}_{\mathcal{T}_{p(n,m)}^{(j)},\textnormal{ }\mathcal{A}}\left(\mathcal{A}(x,\mathcal{T}_{p(n,m)}^{(j)},1^{m})=c_n^{(j)}(x)\right)\geq \frac{2}{3}\right)\geq 1-\frac{1}{m}
    \end{equation}
    is equivalent to 
    \begin{equation}
        \mathbb{P}_{x\sim D_n}\left( \mathbb{P}_{\mathcal{T}_{p(n,m)}^{(j)},\textnormal{ }\mathcal{A}}\left(\mathcal{A}(x,\mathcal{T}_{p(n,m)}^{(j)},1^{m})=c_n^{(j)}(x)\right)\geq 1-\frac{1}{m}\right)\geq 1-\frac{1}{m}
    \end{equation}
    This equivalence follows by using a majority vote. Therefore,
    \begin{align}\label{second_inequality_equivalence}
        \mathbb{E}_{x\sim D_n,\mathcal{T}_{p(n,m)}^{(j)},\textnormal{ }\mathcal{A}}&\left[ \mathbbm{1} \left\{\mathcal{A}(x,\mathcal{T}_{p(n,m)}^{(j)},1^{m})=c_n^{(j)}(x)\right\} \right] \nonumber \\ &= \mathbb{E}_{x\sim D_n}\left[ \mathbb{P}_{\mathcal{T}_{p(n,m)}^{(j)},\textnormal{ }\mathcal{A}} \left\{\mathcal{A}(x,\mathcal{T}_{p(n,m)}^{(j)},1^{m})=c_n^{(j)}(x)\right\} \right] \nonumber \\ & \geq \mathbb{E}_{x\sim D_n}\left[ \mathbb{P}_{\mathcal{T}_{p(n,m)}^{(j)},\textnormal{ }\mathcal{A}} \left\{\mathcal{A}(x,\mathcal{T}_{p(n,m)}^{(j)},1^{m})=c_n^{(j)}(x)\right\} |x \in \mathcal{X}_{{1/m}}(n,m)\right] \nonumber \\ &\hspace{1cm} \times \mathbb{P}_{x\sim D_n}( x \in \mathcal{X}_{{1/m}}(n,m)) \nonumber \\ &\geq \left(1-\frac{1}{m}\right)^2 \geq 1-\frac{2}{m}.
    \end{align}    
    where $\mathcal{X}_{1/m}(n,m)$ denotes the set of Boolean strings $x\in \{0,1\}^n$ that satisfy    
    \begin{equation}
        \mathbb{P}_{\mathcal{T}_{p(n,m)}^{(j)},\textnormal{ }\mathcal{A}}\left(\mathcal{A}(x,\mathcal{T}_{p(n,m)}^{(j)},1^{m})=c_n^{(j)}(x)\right)\geq 1-\frac{1}{m}
    \end{equation}
    Inequality \eqref{second_inequality_equivalence} proves that if a concept class $\mathcal{C}$ is average-case learnable for sequence of distributions $D$ following Definition \ref{aux_definition}, then it is also average-case learnable according to Definition \ref{average_case}. Therefore, by combining this result with the previous one, we conclude that the two definitions are equivalent.
\end{proof}

\section{Proof of Lemma \ref{lemma_2}}\label{appendix_A}

In this appendix, we provide the proof of Lemma \ref{lemma_2}, restated here for completeness.
    \begin{lemma*}
    The concept class $\mathcal{C}_n=\{c_n^{(j)}(x)=\mathbbm{1}\{x=j-1\}\}_{j=1}^{2^n}$ satisfies:

    \begin{enumerate}[$(i)$]
        \item $\mathcal{C}$ is not classically worst-case learnable.
             
        \item All Boolean functions $c_n^{(j)}$ can be efficiently computed by a classical algorithm.
    
        \item $|\mathcal{C}_n|= 2^{n}$ for all $n\in\mathbb{N}$.
    \end{enumerate}
    \end{lemma*}

\begin{proof}
    Conditions $(ii)$ and $(iii)$ follow trivially from the definition. Regarding condition $(i)$, note that unless  a pair $(x,1)$ appears in the training set, it is impossible to guess the correct concept with a probability greater than $1/(2^n-p(n))$. 

    Formally, for any algorithm $A$ (even those that can run in exponential time),
    \begin{align}
        p_e(A):&= \mathbb{E}_{c_n^{(j)}}\mathbb{E}_{\mathcal{T}_{p(n)}^{(j)}} \left[ \mathbb{P}\left(A(\mathcal{T}_{p(n)}^{(j)})\neq j \right)\right] \nonumber \\ & \geq \mathbb{E}_{c_n^{(j)}}\mathbb{E}_{\mathcal{T}_{p(n)}^{(j)}} \left[1- \max_{i} \mathbb{P}\left( c^{(i)}_n | \mathcal{T}_{p(n)}^{(j)}\right) \right] 
    \end{align}
    where the prior probability of each concept $c_n^{(j)}$ is assumed to be equal to $1/2^n$. In case there is a pair $(x,1)$ in the training set, then $\mathbb{P}\left( c^{(i)}_n | \mathcal{T}_{p(n)}^{(j)}\right)=\delta_{i,j}$, otherwise 
    \begin{equation}
        \mathbb{P}\left( c_n^{(i)} | \mathcal{T}_{p(n)}^{(j)}\right)= \left\{\begin{matrix}
        0 \text{    if }(i-1,0)\in \mathcal{T}_{p(n)}^{(j)} \\
        \frac{1}{2^n-p(n)} \text{ otherwise}
        \end{matrix}\right.
    \end{equation}
    Hence, 
    \begin{align}
        p_e(A) &\geq \mathbb{E}_{c_n^{(j)}} \left[\left(1-\frac{1}{2^n-p(n)}\right) \left(1-\mathbb{P}_{x\sim D_n}(j-1)\right)^{p(n)}\right] \noindent \\& \geq \left(1-\frac{1}{2^n-p(n)}\right) \left(1-\frac{1}{2^n}\right)^{p(n)}
    \end{align}
    where the first step uses the fact that the probability of the pair $(j-1,1)$ not appearing in the training set is $\left(1-\mathbb{P}_{x\sim D_n}(j-1)\right)^{p(n)}$, and substitutes $\mathbb{P}(c_n^{(i)}|\mathcal{T}_{p(n)}^{(j)})$. The second inequality uses Jensen's inequality. As
    \begin{equation}
        \lim_{n\rightarrow \infty} \left(1-\frac{1}{2^n}\right)^{p(n)}=1,
    \end{equation}
    for any algorithm, $p_e(A)\rightarrow 1$ as $n\rightarrow \infty$.

    Next, we assume that $\mathcal{C}$ is worst-case learnable. That is, there exists some algorithm $\mathcal{A}$, and polynomial $p(n)$, such that
    \begin{equation}
        \mathbb{P}\left(\mathcal{A}(x,\mathcal{T}_{p(n)}^{(j)})=c_n^{(j)}(x)\right)\geq 2/3
    \end{equation}
    for all $x\in\{0,1\}^n$, $c_n^{(j)} \in \mathcal{C}_n $, and $n\in \mathbb{N}$. If provided with $k$ training sets of size $p(n)$, running the algorithm with each training set and deciding based on a majority vote allows us to exponentially reduce the probability of error. The resulting probability of error is denoted as $c\cdot a^k$. 

    Therefore, the probability that given a training set of size $k\cdot p(n)$ the majority vote algorithm makes an error when listing all values $c_{n}^{(j)}(x)$ for $x\in\{0,1\}^n$ is upper bounded as follows. 
    \begin{equation}
        \mathbb{P}\left( \bigcup_{x\in \{0,1\}^n} {\mathcal{A}}'(x,\mathcal{T}_{kp(n)}^{(j)})\neq c_{n}^{(j)}(x)\right) \leq c \cdot 2^n a^k
    \end{equation}
    Taking $k=O(n\log\frac{1}{\delta})$, the probability of error is bounded by $\delta$. Thus, given a training set of size $O(n\,p(n)\log\frac{1}{\delta})$, we can identify (in exponential time) the concept $c_n^{(j)}$ with a probability of at least $1-\delta$. However, this contradicts the previous result, i.e., $p_e({A})\rightarrow 1$. Therefore, $\mathcal{C}$ is not worst-case learnable.

\end{proof}

\section{Proof of Lemma \ref{lemma_4}}\label{proof_lemma4}

In this appendix, we provide the proof of Lemma \ref{lemma_4}, restated here for completeness.

\begin{lemma*}
    The concept class $\mathcal{C}_n=\{c_n^{(j)}(x)=\mathbbm{1}\{x=j-1\}\}_{i=1}^{2^n}$ satisfies:

    \begin{enumerate}[$(i)$]
        \item $\mathcal{C}$ is not worst-case learnable by a quantum algorithm given quantum sampling advice.

        \item All Boolean functions $c_n^{(j)}$ can be efficiently computed by a classical algorithm.
    
        \item $|\mathcal{C}_n|= 2^{n}$ for all $n\in\mathbb{N}$.
    \end{enumerate}
\end{lemma*}

\begin{proof}
    Similarly to the proof of Lemma \ref{lemma_2}, if $\mathcal{C}_n$ is learnable from a polynomial number of copies of states $\ket{\psi_n^{(i)}}$, then states
    \begin{equation}
        \ket{\psi_n^{(i)}}=\sum_{x\in \{0,1\}^n} \sqrt{\mathbb{P}_{x\sim D_n}(x)}\ket{x}\otimes |c_n^{(i)}(x)\rangle
    \end{equation}
     should be distinguishable. Consequently, states $\ket{\psi_n^{(i)}}$ and $\ket{\psi_n^{(j)}}$, chosen to satisfy  $\left|\bra{\psi_n^{(i)}}\ket{\psi_n^{(j)}}\right|=\max_{p\neq q}\left|\bra{\psi_n^{(p)}}\ket{\psi_n^{(q)}}\right|$, are distinguishable with a polynomial number of samples. The optimal probability of error, given $N$ copies of one of these states, is expressed as
    \begin{equation}
        p_e^*=\max_{i,j:i\neq j}\frac{1}{2}\left(1-\frac{1}{2}\left \|\ket{\psi_n^{(i)}}\bra{\psi_n^{(i)}}^{\otimes N}-\ket{\psi_n^{(j)}}\bra{\psi_n^{(j)}}^{\otimes N}  \right \|_1 \right)
    \end{equation}
    Next, using that $\frac{1}{2}\left \| \rho-\sigma  \right \|_1= \sqrt{1-F(\rho,\sigma)}$ if $\rho$ and $\sigma$ are pure states, and $F(\rho^{\otimes N},\sigma^{\otimes N})=F(\rho,\sigma)^N$,
    \begin{equation}
        p_e^*= \frac{1}{2}\left(1-\sqrt{1- \max_{i,j:i\neq j} \left|\bra{\psi_n^{(i)}}\ket{\psi_n^{(j)}}\right|^{2N}}\right)
    \end{equation}
    Next, substituting states $\ket{\psi_n^{(j)}}$ and $\ket{\psi_n^{(i)}}$, for $i\neq j$,
    \begin{align}
        \bra{\psi_n^{(i)}}\ket{\psi_n^{(j)}}&= \sum_{x,x'} \sqrt{\mathbb{P}_{x\sim D_n}(x)\cdot\mathbb{P}_{x\sim D_n}(x')}\left(\bra{x}\otimes \langle c_n^{(i)}(x)|\right) \left(|\ket{x'}\otimes |c_n^{(j)}(x')\rangle\right) \nonumber \\&=\sum_{x\in \{0,1\}^n}\mathbb{P}_{x\sim D_n}(x)  \langle c_n^{(i)}(x)|c_n^{(j)}(x)\rangle
        =\sum_{x\in \{0,1\}^n}\mathbb{P}_{x\sim D_n}(x)  \bra{\delta_{x,i-1}}\ket{\delta_{x,j-1}} \nonumber \\ &= 1- \mathbb{P}_{x\sim D_n}(i-1)- \mathbb{P}_{x\sim D_n}(j-1)
    \end{align}
    where the third equality follow from the definition of the concepts $c_n^{(j)}(x)$.
    Using this result,
    \begin{align}
        \max_{i,j:i\neq j} \left|\bra{\psi_n^{(i)}}\ket{\psi_n^{(j)}}\right|^{2} &= \left( 1- \min_{i,j:i\neq j} \Big( \mathbb{P}_{x\sim D_n}(i-1)+ \mathbb{P}_{x\sim D_n}(j-1) \Big) \right)^2\nonumber \\ &\geq  \left(1- \frac{2}{2^n} \right)^2
    \end{align}
    The inequality follows from the fact that the sum of the probabilities of the two least probable outcomes in a distribution over $\{0,1\}^n$ is upper-bounded by $2/2^n$. Consequently,
    \begin{equation}
        N\geq \frac{\log \left( 1- (1-2p_e^*)^2 \right)}{2 \log \left( 1- \frac{2}{2^n}\right)} \geq  \frac{-2^n\log \left( 1- (1-2p_e^*)^2 \right)}{4}
    \end{equation}
    which contradicts the statement that the states can be distinguished with a polynomial number of samples. Therefore, the concept class $\mathcal{C}_n$ cannot be learned by a quantum algorithm with quantum sampling advice.

\end{proof}

\section{Proof of Theorem \ref{theorem_worst_case}}\label{proof_theorem3}

In this appendix, we provide the proof of Theorem \ref{theorem_worst_case}, restated here for completeness.

\begin{theorem*}
    The condition $\mathsf{BQP}\, \not\subset \, \mathsf{P/poly}$ is sufficient for the existence of a concept class $\mathcal{C}$ that is (worst-case) learnable by a quantum procedure but is not classically learnable. Additionally, the condition $\mathsf{PromiseBQP} \, \not\subset \, \mathsf{PromiseP/poly}$ is a necessary condition for the existence of a concept class exhibiting the aforementioned property.
\end{theorem*}

\begin{proof}
    First, we prove the sufficient part of the statement. Therefore, we assume $\mathsf{BQP} \not\subset \mathsf{P/poly}$, which implies the existence of a language $L \in \mathsf{BQP}$ such that $L \notin \mathsf{P/poly}$. Then, a concept class formed only by the concept $c_n(x) = \mathbbm{1}\{x \in L\}$ is quantum learnable by definition. However, it is not classically learnable under any distribution since $\mathsf{P/poly} = \mathsf{BPP/samp(All)}$.

    Now, we move to the necessary part. That is, we assume the existence of a concept class $\mathcal{C}$ that is (worst-case) quantum learnable but not classically learnable. This assumption implies the existence of a triplet $(\mathcal{A}_Q, D, p)$ such that
    \begin{equation}
        \mathbb{P}\left( \mathcal{A}_Q(x,\mathcal{T}_{p(n)}^{(j)})=c_n^{(j)}(x)\right)\geq 1-\delta_0
    \end{equation}
    for all $x \in \{0,1\}^n$, $c_n^{(j)}\in \mathcal{C}_n$, and $n\in \mathbb{N}$, where $\mathcal{A}_Q$ is an efficient quantum algorithm, and $\delta_0$ denotes a small constant. Similarly, for any classical algorithm, distribution and polynomial
    \begin{equation}
        \mathbb{P}\left( \mathcal{A}_C(x,\mathcal{T}_{q(n)}^{(j)})=c_n^{(j)}(x)\right)< \frac{1}{2}+\delta_0
    \end{equation}
    for some $x \in \{0,1\}^n$, and $c_n^{(j)}$. 

    From this assumption, the objective is to construct a promise problem $L=(L_{\mathrm{YES}},L_{\mathrm{NO}})$ such that $L \in \mathsf{PromiseBQP}$, and $L \notin \mathsf{PromiseP/poly}$. The proposed problem is motivated by the fact that there is no classical algorithm capable of simulating $\mathcal{A}_Q$ on all inputs. Otherwise, the concept class would be classically learnable. The specific definition is:

    \begin{align}
        L_{\mathrm{YES}}=\bigg\{ \left< x,\mathcal{T}_{p(n)}\right> : \argmax_{b\in\{0,1\}} \left \{  \mathbb{P}(\mathcal{A}_Q(x, \mathcal{T}_{p(n)})=b)\right\}=1 \nonumber \\ \text{ and } \max_{b\in\{0,1\}}  \mathbb{P}(\mathcal{A}_Q(x, \mathcal{T}_{p(n)})=b)\geq 1/2+\epsilon\bigg \}
    \end{align}
    and  
    \begin{align}
        L_{\mathrm{NO}}=\bigg\{ \left< x,\mathcal{T}_{p(n)}\right> : \argmax_{b\in\{0,1\}} \left \{  \mathbb{P}(\mathcal{A}_Q(x, \mathcal{T}_{p(n)})=b)\right\}=0 \nonumber \\ \text{ and } \max_{b\in\{0,1\}}  \mathbb{P}(\mathcal{A}_Q(x, \mathcal{T}_{p(n)})=b)\geq 1/2+\epsilon\bigg \}
    \end{align}
    where the probability is taken with respect to the internal randomness of algorithm $\mathcal{A}_Q$, and $\epsilon$ denotes a small constant. The factor $\max_{b\in\{0,1\}}  \mathbb{P}(\mathcal{A}_Q(x, \mathcal{T}_{p(n)})=b)\geq 1/2+\epsilon$  guarantees that we can distinguish in a reasonable time wheter $\left< x,\mathcal{T}_{p(n)}\right>$ belongs to $L_{\mathrm{YES}}$ or $L_{\mathrm{NO}}$, as long as it belongs to one of them.

    Clearly, $L \in \mathsf{PromiseBQP}$ since running $\mathcal{A}_Q$ multiple times allows us to choose $b \in \{0,1\}$ as the most frequent output, ensuring the correct outcome with high probability.

    Now, we show that $L \notin \mathsf{PromiseP/poly}$. To do so, we use a contradiction argument. That is, we start by assuming that $L \in \mathsf{PromiseP/poly}$, meaning there exists a Turing machine $M$ and a sequence of strings $a_n$ such that $M(x, \mathcal{T}_{p(\ell(x))}, a_{\ell(x)})$ solves problem $L$. Let polynomial $h(n)$ be the length of $a_n$.

    In a similar spirit than in the proof of Theorem 1, we encode advice $a_n$ into the distribution $D_n$ without affecting to many entries. In particular, note that for any distribution $D_n$ over $\{0,1\}^n$, there exists a set of $S\subset\{0,1\}^n$ with cardinality $h(n)$ such that 
    \begin{equation}
        \sum_{x\in S} \mathbb{P}_{X\sim D_n}(x) \leq \frac{h(n)}{2^n}
    \end{equation}
    otherwise $\sum_{x} \mathbb{P}_{X\sim D_n}(x) > 1$. The idea of the encoding is simple: first, we order the elements of $S$ by interpreting the strings of $n$ bits as integer values. Starting with the smallest value in $S$, we encode the first bit of $a_n$ by increasing the probability by $p$ or $p+p_1$ depending on the bit, and then proceed to the second bit and so on. The probability $p$ is introduced to allow the detection of the set $S$, i.e., the positions containing the bits of advice $a_n$. The encoding makes the following transformation to distribution $D_n$,
    \begin{equation}
        \mathbb{P}_{X\sim D_{\mathrm{adv},n}}(x) =\frac{1}{C}\left( \frac{1}{3h(n)}\mathbb{P}_{X\sim D}(x)+\mathbbm{1}(x\in S) \frac{2+b(x)}{3h(n)}\right)
    \end{equation}
    where $b(x)$ denotes the specific bit that has to be encoded into string $x\in S$, and $C$ is a normalization constant, which satisfies that $2/3\leq C\leq 4/3$. The particular value of $C$ depends on the Hamming weight of $a_n$. 
    
    \smallskip
    \noindent Now, we prove three properties of this distribution:

    \begin{enumerate}[$(i)$]
        \item There exists a polynomial $q_1(n)$ such that sampling $q_1(n)$ times the distribution $D_{\mathrm{adv},n}$ allows the advice $a_n$ to be decoded with high probability.  
        \item There exists a polynomial $q_2(n)$ such that sampling $q_2(n)$ times the distribution $D_{\mathrm{adv},n}$, we obtain with high probability at least $p(n)$ samples where $x \notin S$.
        \item The total variation between distributions $\mathbb{P}_{X\sim D_n}(x)$ and $\mathbb{P}_{X\sim D_{\mathrm{adv},n}}(x|x\notin S)$ decays exponentially. 
        
    \end{enumerate}

\smallskip\smallskip
    To prove property $(i)$, Hoeffding's inequality is used, which states that
    \begin{equation}
        \mathbb{P}\left(\left| \frac{1}{N}\sum_{i=1}^N{\mathbbm{1}\{x_i=z\}} - \mathbb{P}_{X\sim D_\mathrm{adv}}(z) \right|\leq \epsilon\right)\geq 1-\delta
    \end{equation}
    for $N=\frac{1}{2\epsilon^2} \log\frac{2}{\delta}$. In particular, we take $\epsilon=  \left(\frac{3}{24 h(n)}\right)^2 \leq \left(\frac{1}{6C \,h(n)}\right)$. This guarantees that we can differentiate (for a fix value $x$) whether $x\in S$ or $x\notin S$. Furthermore, if $x\in S$, we can distinguish the encoded bit. Next, using the union bound,
    \begin{equation}
        \mathbb{P}\left(\bigcup_{z\in\{0,1\}^n}\left| \frac{1}{N}\sum_{i=1}^N{\mathbbm{1}\{x_i=z\}} - \mathbb{P}_{X\sim D_\mathrm{adv}}(z) \right|> \epsilon\right)< 2^n\delta
    \end{equation}
    Hence, by setting $\delta=\delta'/2^n$, we ensure that with probability at least $1-\delta'$, all empirical probabilities deviate by at most $\epsilon$. Consequently, with probability at least $1-\delta'$, we can correctly decode the advice $a_n$. Note that after substituting the values of $\epsilon$ and $\delta$, the number of samples $N$ remains a polynomial in $n$.

    For property $(ii)$, we conclude that $N=16\, p(n) h(n)$ samples from $D_{\mathrm{adv},n}$ are sufficient. To prove this, we use Chebyshev's inequality, with the random variable $Z:=\sum_{i=1}^N{\mathbbm{1}\{x_i\notin S\}}$,
    \begin{equation}
        \mathbb{P}\left(\left| Z - \mathbb{E} [Z] \right| \geq  K \sigma \right)\leq \frac{1}{K^2}
    \end{equation}
    Since $\mathbbm{1}\{x_i\notin S\}$ follows a Bernoulli distribution with parameter $\mathbb{P}_{X\sim D_{\mathrm{adv},n}}(x\notin S)$,
    \begin{align}
    \mathrm{Var}[Z] &= N \mathbb{P}_{X\sim D_\mathrm{adv}}(x\notin S) (1-\mathbb{P}_{X\sim D_\mathrm{adv}}(x\notin S)) \leq  N \mathbb{P}_{X\sim D_\mathrm{adv}}(x\notin S) \nonumber \\ & \leq \frac{N}{2 h(n)}=8p(n)
    \end{align}
    where for the second inequality we use that $\mathbb{P}_{X\sim D_\mathrm{adv}}(x\notin S)\leq \frac{1}{2 h(n)}$. Therefore,
    \begin{equation}
        \mathbb{P}\left(\left| Z - \mathbb{E} [Z] \right| \geq K \sqrt{8p(n)} \right)\leq \frac{1}{K^2}
    \end{equation}
    Next, taking $K=\sqrt{p(n)/8}$
    \begin{equation}\label{eq_bound_Z}
        \mathbb{P}\left(\left| Z - \mathbb{E} [Z] \right| \geq  p(n) \right)\leq \frac{8}{p(n)}
    \end{equation}
    Now, we lower bound $\mathbb{P}_{X\sim D_{\mathrm{adv},n}}(x\notin S)$,
    \begin{align}
        \mathbb{P}_{X\sim D_{\mathrm{adv},n}}(x\notin S) &= \frac{1}{3Ch(n)}\mathbb{P}_{X\sim D_n}(x\notin S) \nonumber \\ &\geq \frac{1}{3Ch(n)} \left( 1-\frac{h(n)}{2^n} \right) \nonumber \\ & \geq \frac{1}{4h(n)} \left( 1-\frac{h(n)}{2^n} \right)\geq \frac{1}{8 h(n)}
    \end{align}
    where the first inequality follows from the definition of $S$, the second one from bounding $C$, and the last one holds for a sufficiently large $n$. This inequality implies that $\mathbb{E}[Z]=N \mathbb{P}_{X\sim D_{\mathrm{adv},n}}(x\notin S) \geq 2p(n)$, which together with \eqref{eq_bound_Z}, 
    \begin{equation}
        \mathbb{P}\left(Z \geq  p(n) \right)\geq 1- \frac{8}{p(n)}
    \end{equation}
    which ends the proof of property $(ii)$.

    Next, we prove the last property, i.e., we bound the total variation between $\mathbb{P}_{X\sim D_n}(x)$ and $\mathbb{P}_{X\sim D_{\mathrm{adv},n}}(x|x\notin S)=\mathbb{P}_{X\sim D_n}(x|x\notin S)$, where the equality follows from the definition of distribution $D_{\mathrm{adv},n}$. For simplicity, we denote the latter distribution as $\tilde{D}_n$.
    \begin{align}
        2 d_{TV}(D_n,\tilde{D}_n ) &= \sum_{x\in\{0,1\}^n}\left|\mathbb{P}_{X\sim D_n}(x)-\mathbb{P}_{X\sim \tilde{D}_n}(x)\right|\nonumber \\& = \left(\frac{1}{\mathbb{P}_{X\sim D_n}(x\notin S)} -1 \right)\mathbb{P}_{X\sim D_n}(x\notin S)+\mathbb{P}_{X\sim D_n}(x\in S) \nonumber \\ & \leq \left(\frac{h(n)}{2^n-h(n)} \right)+ \frac{h(n)}{2^n} \leq \frac{2 \,h(n)}{2^n-h(n)} 
    \end{align}
    where the inequality follows from $\mathbb{P}_{X\sim D_n}(x\in S)\leq h(n)/2^n$. To distinguish between these distribution with an average probability of error $p_e$ (the prior is 1/2 for each hypothesis), we need at least 
    $N \geq \frac{-4 p_e+2}{ d_{TV}(D_n,\tilde{D}_n )}$
    samples \cite{lehmann1986testing}. The lower bound grows exponentially with $n$. This implies that when using distribution $\tilde{D}_n$, 
    \begin{equation}
        \mathbb{P}\left( \mathcal{A}_Q(x,\mathcal{T}_{p(n)}^{(j)})=c_n^{(j)}(x)\right)\geq 1-\delta_0-\Delta
    \end{equation}
    for all $x \in \{0,1\}^n$, $c_n^{(j)}\in \mathcal{C}_n$, and $n\in \mathbb{N}$, where $\Delta$ denotes a small constant. To prove this, let's assume that the latter inequality does not hold. Then, there exist $x$, and $c_n^{(j)}$ such that 
    \begin{equation}
        \mathbb{P}\left( \mathcal{A}_Q(x,\mathcal{T}_{p(n)}^{(j)})=c_n^{(j)}(x)\right)<1-\delta_0-\Delta
    \end{equation}
    Therefore, since when the training set is distributed according to $D_n$, the algorithm $\mathcal{A}_Q$ satisfies,
    \begin{equation}
        \mathbb{P}\left( \mathcal{A}_Q(x,\mathcal{T}_{p(n)}^{(j)})=c_n^{(j)}(x)\right)\geq 1-\delta_0
    \end{equation}
    we could distinguish both distributions with $O(1/\Delta^2)$ repetitions, i.e., using $O(p(n)/\Delta^2)$ samples. This is contradiction with the previous result.

    In a nutshell, there exists a sufficiently large polynomial $q(n)$ such that: $(i)$ We can reliably decode the advice $a_n$, and $(ii)$ the number of samples satisfying $x\notin S$ is greater than $p(n)$ with a high probability. Additionally, $\mathcal{A}_Q$ successfully learns the concept class $\mathcal{C}$ using the sequence of distributions $\tilde{D}$. 

    \smallskip\smallskip
    
    \noindent
    Now, we propose a classical algorithm that learns the concept class $\mathcal{C}$:
    
    \begin{itemize}
        \item The number of samples in the training set is set to the polynomial $q(n)$, and the distribution used is $D_{\mathrm{adv,n}}$. The polynomial $q(n)$ is chosen to be sufficiently large, ensuring reliable decoding of the advice with high probability, and also guaranteeing with high probability that there are more than $p(n)$ samples for which $x\notin S$.

        \item Next, from the training set $\mathcal{T}_{q(n)}^{(j)}(D_{\mathrm{adv},n})$, we extract the advice $a_n$ and form a new training set $\mathcal{T}_{p(n)}^{(j)}(\tilde{D}_{n})$.

        \item Finally, we execute the Turing machine $M$ with input $(x, \mathcal{T}_{p(n)}^{(j)}(\tilde{D}_{n}), a_n)$, which yields the most probable output of $\mathcal{A}_Q(x, \mathcal{T}_{p(n)}^{(j)}(\tilde{D}_{n}))$. By its definition, it outputs $c_n^{(j)}(x)$ with probability for any $x$, $c_n^{(j)}$, and $n\in\mathbb{N}$.
    \end{itemize}

    \noindent
    Therefore, a contradiction arises, as this would imply that the concept class is classically worst-case learnable. Consequently, $L\notin \,\mathsf{PromiseP/poly}$.

\end{proof}

\section{Proof of Theorem \ref{theorem_average_case}}\label{proof_average}

In this section, we prove Theorem \ref{theorem_average_case}. First, we start by introducing the formal definitions of $\mathsf{HeurBPP/samp}$($D$) and $\mathsf{HeurBQP/samp}$($D$).

\begin{definition}
    $\mathsf{HeurBPP/samp}$\textnormal{(${D}$)} is the class of distributional problems $(L,{D}')$ that satisfy:
    \begin{enumerate}[$(i)$]
        \item ${D}'$ equals ${D}$,
        
        \item there exists a probabilistic TM $M$ running in polynomial time and a polynomial $p(n,m)$, such that for the fixed sequence of distributions ${D}=\{D_n\}_{n\in \mathbb{N}}$, 
    \begin{equation}
        \mathbb{P}_{x\sim D_n} \left(\mathbb{P}_{\mathcal{T}_{p(\ell(x),m)},\, M}\left(M(x,\mathcal{T}_{p(\ell(x),m)},1^m)=\mathbbm{1}\{x\in L\}\right)\geq \frac{2}{3} \right)\geq 1-\frac{1}{m}
    \end{equation}
    for any $n,m\in \mathbb{N}$, where $\mathcal{T}_{p(n,m)}=\{(x_i,\mathbbm{1}\{x_i\in L\})\}_{i=1}^{p(n,m)}$, and $x_i\sim D_n$. The external probability is taken with respect to $x\sim D_n$, and the internal probability with respect to the training set $\mathcal{T}_{p(n,m)}$, and the internal randomness of  TM $M$.
    \end{enumerate}
\end{definition}

The definition of $\mathsf{HeurBQP/samp}$($D$) is analogous to that of $\mathsf{HeurBPP/samp}$($D$), with the probabilistic Turing machine $M$ replaced by a quantum algorithm running in polynomial time. With these definitions introduced, we can proceed with the proof.

\begin{theorem*}
    The condition $\mathsf{HeurBPP/samp}(D)\,\subsetneq\,\mathsf{HeurBQP/samp}$$(D)$ is a necessary and sufficient condition for the existence of a concept class $\mathcal{C}$ that is average-case learnable by a quantum procedure for sequence $D$ but it is not classically average-case learnable for sequence $D$.
\end{theorem*}

\begin{proof}
    First, we prove the sufficient part of the statement. Therefore, we assume $\mathsf{HeurBPP/samp}(D)\,\subsetneq\,\mathsf{HeurBQP/samp}$$(D)$, which implies the existence of a distributional problem $(L,D) \in \mathsf{HeurBQP/samp}(D)$ such that $(L,D)$ does not belong to $\mathsf{HeurBPP/samp}$$(D)$. Then, a concept class formed only by the concept $c_n(x) = \mathbbm{1}\{x \in L\}$ is average-case quantum learnable for sequence $D$ and not classically learnable. 

    For the necessary part, we start by assuming that exists some concept class $\mathcal{C}$ such that it is average-case quantum learnable for distribution $D$ but not classically learnable. As in Lemma \ref{lemma_1}, we are going to transform this concept class into a language $L(\mathcal{C}):=\{x:c_{\ell(x)}^{\left(g(\ell(x))\right)}(x)=1\}$. For the definition of function $g(n)$, first, we need to modify slightly the notation used in Lemma \ref{lemma_1}. Specifically, given an algorithm $\mathcal{A}$, a sequence of distributions $D$, and a polynomial $p(n,m)$, we define the set $\tilde{\mathcal{E}}_n(\mathcal{A},D,p)\subseteq \mathcal{C}_n$, as the set of concepts $c_n^{(j)}$  that satisfy
    \begin{equation}
            \mathbb{P}_{x\sim D_n}\left( \mathbb{P}_{\mathcal{T}_{p(n,m)},\textnormal{}\mathcal{A}}\left(\mathcal{A}(x,\mathcal{T}_{p(n,m)}^{(j)},1^{m})=c_n^{(j)}(x)\right)\geq \frac{2}{3}\right)< 1-\frac{1}{m}
    \end{equation}   
    for some $m\in \mathbb{N}$. Therefore, if $\mathcal{C}$ is not classically average-case learnable for sequence $D$, then any pair $(\mathcal{A},p)$ satisfies that $\mathcal{E}_n(\mathcal{A},D,p)\neq\emptyset$ for some values $n\in \mathbb{N}$. The set of values of $n\in \mathbb{N}$ for which this occurs is denoted by
    \begin{equation}
        \tilde{\mathcal{N}}(\mathcal{A},D,p):=\{n\in \mathbb{N}: \tilde{\mathcal{E}}_n(\mathcal{A},D,p)\neq \emptyset\}
    \end{equation}
    Let $B_{\mathrm{alg}}$ denote a bijection between $\mathbb{N}$ and the set of efficient classical algorithms, and $B_{\mathrm{poly}}$ represent a bijection between $\mathbb{N}$ and the set of polynomials $\{c\,n^{k_1} m^{k_2}:(c,k_1,k_2)\in\mathbb{N}^3\}$. Now, we can define function $g(n)$:
    \begin{equation}
         g(n):=\left\{\begin{matrix}1 \text{ if }n\notin\{n_i\}_{i=1}^{\infty}\\ \min \{j: c_n^{(j)}\in \tilde{\mathcal{E}}_n\left(B_{\mathrm{alg}}(s_{i^*(n)}),D,B_{\mathrm{poly}}(z_{i^*(n)})\right)\}\text{ if }n\in\{n_i\}_{i=1}^{\infty}
    \end{matrix}\right.
    \end{equation}
    where sequences $s_i\in \mathbb{N}$ and $z_i\in \mathbb{N}$ satisfy that any point $(a,b)\in \mathbb{N}^2$ appears at least once in the sequence $(s_i,z_i)$, $n_1:=\min\{\tilde{\mathcal{N}}\left(B_{\mathrm{alg}}(s_{1}),D,B_{\mathrm{poly}}(z_{1})\right) \}$, and $n_i=\min\{\tilde{\mathcal{N}}\left(B_{\mathrm{alg}}(s_{i}),D,B_{\mathrm{poly}}(z_{i})\right) \backslash\{n_1,n_2,\cdots,n_{i-1}\} \}$. Finally, $i^*(n)$ denotes the index $i$ such that $n_i=n$. As in Lemma \ref{lemma_1}, this construction of $L(\mathcal{C})$ implies that the distributional problem $(L(\mathcal{C}),D)$ is not in $\mathsf{HeurBPP/samp}$$(D)$. However, since the concept class is quantum average-case learnable for sequence $D$, then $(L(\mathcal{C}),D) \in \mathsf{HeurBQP/samp}$$(D)$.

\end{proof}

\end{appendices}

\end{document}